\documentclass[runningheads]{llncs}
\usepackage{graphicx}
\usepackage{todonotes}
\usepackage{amsmath}
\usepackage{enumerate} 
\usepackage{xspace}
\usepackage{subfig}
\usepackage{tikz}
\usepackage{hyperref}

\newcommand{\reductionAvg}{$43.8$\xspace}

\newcommand{\reductionMedian}{$40.0$\xspace}

\newcommand{\numIterationsMaxAvg}{$3.8$\xspace}
\newcommand{\excludedFacesAvg}{$76.5$\xspace}
\newcommand{\orthogonalZeroBends}{$129$\xspace}
\newcommand{\orZeroBends}{$550$\xspace}
\newcommand{\maxNumVertices}{$44$\xspace}
\newcommand{\minNumVertices}{$3$\xspace}

\newcommand{\timeLessHalfSecond}{$1081$\xspace}
\newcommand{\timeMoreTenSeconds}{$861$\xspace}
\newcommand{\timeMoreOneMinute}{$628$\xspace}

\newcommand{\numInstances}{$4048$\xspace}
\newcommand{\timeOut}{$586$\xspace}

\newcommand{\optimallySolved}{$3462$\xspace}

\DeclareMathOperator{\rot}{rot}

\newcommand{\eref}{\ensuremath{{e^\star}}}

\newcommand*\patchAmsMathEnvironmentForLineno[1]{%
  \expandafter\let\csname old#1\expandafter\endcsname\csname #1\endcsname
  \expandafter\let\csname oldend#1\expandafter\endcsname\csname end#1\endcsname
  \renewenvironment{#1}%
     {\linenomath\csname old#1\endcsname}%
     {\csname oldend#1\endcsname\endlinenomath}}%
\newcommand*\patchBothAmsMathEnvironmentsForLineno[1]{%
  \patchAmsMathEnvironmentForLineno{#1}%
  \patchAmsMathEnvironmentForLineno{#1*}}%
\AtBeginDocument{%
\patchBothAmsMathEnvironmentsForLineno{equation}%
\patchBothAmsMathEnvironmentsForLineno{align}%
\patchBothAmsMathEnvironmentsForLineno{flalign}%
\patchBothAmsMathEnvironmentsForLineno{alignat}%
\patchBothAmsMathEnvironmentsForLineno{gather}%
\patchBothAmsMathEnvironmentsForLineno{multline}%
}

\renewcommand{\orcidID}[1]{\href{https://orcid.org/#1}{\includegraphics[scale=.03]{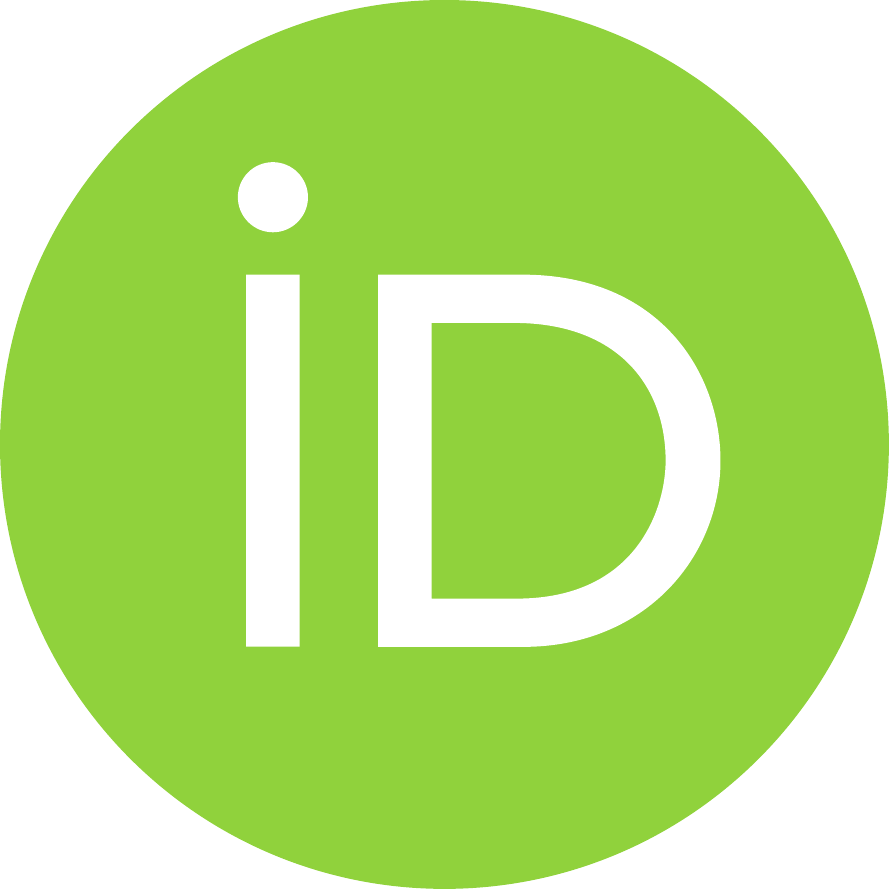}}}

\begin{document}
\title{An Integer-Linear Program for Bend-Minimization in Ortho-Radial
  Drawings}

\titlerunning{ILP for Ortho-Radial Drawings}

\author{Benjamin Niedermann\inst{1}\orcidID{0000-0001-6638-7250} \and
  Ignaz Rutter\inst{2}\orcidID{0000-0002-3794-4406}}
\authorrunning{B. Niedermann \and I. Rutter}

\institute{Universit{\"a}t Bonn, 53115 Bonn, Germany, \email{niedermann@uni-bonn.de} 
  \and Universit{\"a}t Passau, 94032 Passau, Germany,
  \email{rutter@fim.uni-passau.de}}
\maketitle              %
\begin{abstract}
  An ortho-radial grid is described by concentric circles and straight-line spokes emanating from the circles' center.  An
  ortho-radial drawing is the analog of an orthogonal drawing on an
  ortho-radial grid.  Such a drawing has an unbounded outer face and a
  central face that contains the origin.  Building on the 
  notion of an ortho-radial representation~\cite{bnrw-ttsmford-17}, we describe an
  integer-linear program (ILP) for computing bend-free ortho-radial
  representations with a given embedding and fixed outer and central
  face.  Using the ILP as a building block, we introduce a pruning
  technique to compute bend-optimal ortho-radial drawings with a given
  embedding and a fixed outer face, but freely choosable central face.
  Our experiments show that, in comparison with orthogonal drawings
  using the same embedding and the same outer face, the use of
  ortho-radial drawings reduces the number of bends
  by~$\reductionAvg\%$ on average. Further, our approach allows us to
  compute ortho-radial drawings of embedded graphs such as the metro
  system of Beijing or London within seconds.

\end{abstract}

\begin{figure}[t]
  \centering
   \begin{picture}(120,100)
     \put(0,0){\includegraphics[page=1]{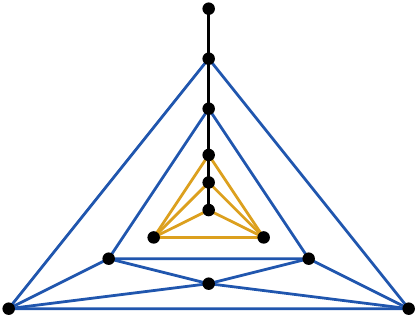}}
     \put(0,90){a)}
   \end{picture}\hfil
   \begin{picture}(100,100)
     \put(0,0){\includegraphics[page=3]{fig/example-biedl-kant}}
     \put(-10,90){b)}
   \end{picture}\hfil
   \begin{picture}(100,100)
     \put(0,0){\includegraphics[page=2]{fig/example-biedl-kant}}
     \put(0,90){c)}
   \end{picture}
   \caption{Graph with (a)~straight-line, (b)~orthogonal and
     (c)~ortho-radial layout. The ortho-radial layout has been created
     by our approach and has 14 bends. The orthogonal layout has been
     proposed by Biedl and Kant~\cite{BiedlK94} and has 23 bends.}
  \label{fig:example:biedl-kant}
\end{figure}

\section{Introduction}
\label{sec:introduction}

Planar orthogonal drawings are arguably one of the most popular
drawing styles.  Their aesthetic appeal derives from their good
angular resolution and the restriction to only the horizontal and the
vertical slope, which makes it easy to trace the edges.  They
naturally correspond to embeddings into the standard grid, where edges
are mapped to paths between their endpoints.  The most important
aesthetic criterion for orthogonal drawings is the number of bends.
Consequently, a large body of literature deals with optimizing the
number of
bends~\cite{t-emn-87,Blasius2016,Blasius2016b,Felsner2014,dg-popda-13,Biedl1998}.
\emph{Ortho-radial drawings} are a natural analog of orthogonal
drawings but on an \emph{ortho-radial grid}, which is formed by
concentric circles and straight-line spokes emanating from the
circles' center.  Besides their aesthetic appeal and the fact that
they inherit favorable properties of orthogonal drawings like a good
angular resolution, they have the potential to save on the
number of bends; see Fig.~\ref{fig:example:biedl-kant}.

The corner-stone of the whole theory of bend minimization is the
notion of an \emph{orthogonal representation}, which for a \emph{plane graph}
(i.e., a graph with a fixed embedding) describes for each vertex the
angles between consecutive incident edges and for each edge the order
and directions of its bends.  It is a seminal result of
Tamassia~\cite{t-emn-87} that characterizes the orthogonal
representations in terms of local conditions and shows that every
orthogonal representation admits a drawing.  The usefulness of this
result hinges on the fact that it turns the seemingly geometric
problem of computing a bend-optimal drawing into a purely
combinatorial one.  Geometric aspects of the drawing, such as choosing
edge lengths, can then be dealt with separately, and after deciding
the bends on the edges.

Today, this is usually described as a pipeline consisting of three
steps, the topology-shape-metrics framework (TSM for short).  The
topology step chooses a planar embedding of the input graph.  The
shape step computes an orthogonal representation for this embedding
(e.g., using flow-based methods), and the metrics step computes edge
lengths so that a crossing-free drawing is obtained.

Recently, this framework has been adapted to ortho-radial drawings.
There is a natural analog of ortho-radial representation that
satisfies analogous local conditions to orthogonal representations.
However, unlike the orthogonal case, there exist ortho-radial
representations that satisfy all the local conditions, but do not
correspond to an ortho-radial drawing; see Fig.~\ref{fig:examples}a,b
for an example.  After initial results on the characterization of
ortho-radial representations of cycles~\cite{hht-orthoradial-09} and
ortho-radial representations of maxdeg-3 graphs, where all faces are
rectangles~\cite{hhmt-rrdcp-10}, Barth et al.~\cite{bnrw-ttsmford-17}
 gave a characterization of the drawable ortho-radial
representations in terms of a third, more global condition.
Niedermann et al.~\cite{nrw-eaorg-19} further showed that, given an
ortho-radial representation that satisfies the third condition, its
ortho-radial drawing can be computed in quadratic time.

Up to now, however, there are no algorithms for computing ortho-radial
representations, even if the graph comes with a fixed planar
embedding, including the central and the outer face.  It is an open
question whether a bend-optimal valid ortho-radial representation can be computed
efficiently in this setting.  The example from
Fig.~\ref{fig:examples}a,b already shows that such an ortho-radial
representation cannot be characterized in terms of purely local
conditions.  Fig.~\ref{fig:examples}c is an example of a
bend-optimal drawing where an edge bends in two different directions.
This shows that a straightforward adaption of existing techniques that
are based on min-cost flows, is unlikely to succeed.  In this paper,
we develop a method for computing ortho-radial representations with
few bends based on an integer-linear program (ILP).  This yields the
first practical algorithm that takes an arbitrary plane maxdeg-4 graph as
input and computes an ortho-radial drawing.  We use it to
evaluate the usefulness of ortho-radial drawings, in
particular with respect to the potential of saving bends in comparison
to orthogonal drawings.

\begin{figure}[tb]
  \centering
  \includegraphics{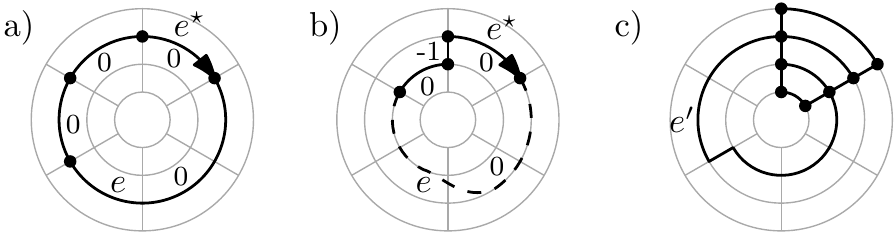}
  \caption{An ortho-radial drawing of a 4-cycle (a) and an
    ortho-radial representation of it that is not drawable (b), though
    the sum of the angles around each vertex and around each face is
    the same as in (a).  A bend-optimal drawing of a graph (c), where
    the edge $e'$ has bends in different directions.}
  \label{fig:examples}
\end{figure}

\paragraph{Contribution and Outline.} We start with preliminary
results in Section~\ref{sec:preliminaries} introducing notions and
facts on orthogonal and ortho-radial representations.  In
Section~\ref{sec:bend-free}, we present an ILP for computing bend-free
ortho-radial representations for graphs with a fixed embedding.  In
Section~\ref{sec:with-bends} we extend that ILP to optimize the number
of bends. To that end, we provide theoretical insights into the number
of bends required for ortho-radial drawings. Moreover, we describe a
pruning strategy that allows us to quickly compute a bend-optimal
drawing.  In Section~\ref{sec:evaluation} we evaluate our algorithms
and compare them to standard approaches for computing orthogonal
drawings.

\section{Preliminaries}
\label{sec:preliminaries}

A graph of maximum degree~4 is a \emph{4-graph}.  Unless stated otherwise,
all graphs occurring in this paper are 4-graphs.  Let $G=(V,E)$ be a
connected planar 4-graph with a fixed combinatorial embedding~$\mathcal E$ and
let~$v \in V$ be a vertex.  We call the counterclockwise order of edges
around~$v$ in the embedding the \emph{rotation of $v$}, and we denote
it by~$\mathcal E(v)$.  An \emph{angle at~$v$} is a pair of
edges~$(e_1,e_2)$ that are both incident to~$v$ and such that~$e_1$
immediately precedes~$e_2$ in~$\mathcal E(v)$.

\begin{figure}[t]
  \centering
  \includegraphics[page=1]{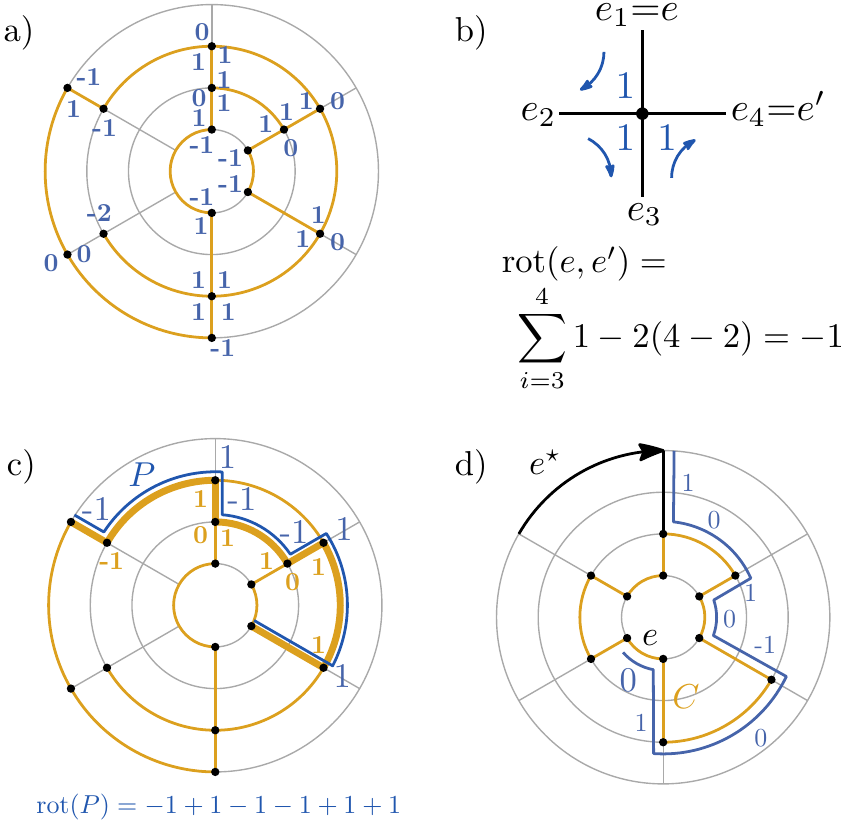}
  \caption{(a) Rotations between angles. (b)~Rotations
    between two edges $e,e'$. (c) Rotation of path $P$. (d) Label of edge $e$ with respect to~essential cycle $C$. }
  \label{fig:prelim}
\end{figure}

Let~$\Delta$ be an orthogonal (or ortho-radial) drawing of $G$ with
embedding~$\mathcal E$.  By turning bends into vertices, we can assume
that the drawing is bend-free.  We now derive a labeling of the angles
of~$v$ with labels in~$\{-2,-1,0,1\}$ with so-called rotation values.
For an angle~$(e_1,e_2)$ at~$v$, we
set~$\rot(e_1,e_2) = 2-2\alpha/\pi$, where $\alpha$ is the
counterclockwise geometric angle between~$e_1$ and~$e_2$ in~$\Delta$; see Fig.~\ref{fig:prelim}a.
Intuitively, this counts the number of right turns one takes when
traversing $e_1$ towards $v$ and afterwards $e_2$ away from $v$, where
negative numbers correspond to left turns.  Note that if~$e_1=e_2$,
then~$\rot(e_1,e_2) = -2$, i.e., $v$ contributes two left turns.

For a face~$f$ of $G$, we denote by~$\rot(f)$ the sum of the rotations
of all angles incident to $f$.  Formally, if~$v_0,\dots,v_{n-1}$ is
the facial walk around~$f$ (oriented such that $f$ lies to the right
of the facial walk), we
define~$\rot(f) = \sum_{i=0}^{n-1} \rot(v_{i-1}v_i,v_iv_{i+1})$ where
indices are taken modulo $n$.  Intuitively, this counts the number of
right turns minus the number of left turns one takes when traversing
the face boundary such that the face $f$ lies to the right.
Since~$\Delta$ is an orthogonal drawing with some outer face~$f_o$, it
satisfies the following conditions~\cite{t-emn-87}.
\begin{enumerate}[(I)]
\item For each vertex, the sum of the rotations around~$v$ is~$2(\deg(v)-2)$.

\item For each face $f \ne f_0$ it is $\rot(f) = 4$ and it
  is~$\rot(f_0) = -4$.
\end{enumerate}

We call an assignment~$\Gamma$ of rotation values to the angles that
satisfy these two rules an \emph{orthogonal representation}.  Every
orthogonal drawing~$\Delta$ induces an orthogonal representation.  An
orthogonal representation~$\Gamma$ is \emph{drawable} if there exits a
drawing~$\Delta$ that induces it. 

For ortho-radial drawings a similar situation occurs.  Here, we have
two special faces; an unbounded face, called the \emph{outer face}
$f_o$, and a \emph{central face} $f_c$, which contains the origin.  If
the central and the outer face are identical, then the drawing does
not enclose the origin, and the ortho-radial drawing does in fact lie
on some patch of the ortho-radial grid that can be conformally mapped
to an orthogonal grid (i.e., without changing any angles).  
An ortho-radial drawing~$\Delta$ similarly defines rotation values
that satisfy the following conditions.
\begin{enumerate}[(I)']
\item For each vertex, the sum of the rotations around~$v$ is~$2(\deg(v)-2)$.\label{condition:or:vertex}
\item For each face $f \ne f_o,f_c$ it is $\rot(f) = 4$, if
  $f_o \ne f_c$, then~$\rot(f_o) = \rot(f_c) = 0$ and~$\rot(f_o) = -4$
  if~$f_o = f_c$.\label{condition:or:face}
\end{enumerate}

Similar to the orthogonal case, an ortho-radial drawing $\Delta$
therefore induces an \emph{ortho-radial representation}~$\Gamma$ that
defines rotation values satisfying these conditions.  An ortho-radial
representation is \emph{drawable} if there exists an ortho-radial
drawing that induces it.

Tamassia~\cite{t-emn-87} proved that every
orthogonal representation is drawable.
In contrast, there exist
ortho-radial representations that are not drawable; see e.g.,
Fig.~\ref{fig:examples}b, which illustrates a so-called strictly monotone cycle.
Its ortho-radial representation satisfies conditions (I)' and (II)',
yet it is not drawable.

To characterize the drawable ortho-radial representations, Barth et
al.~\cite{bnrw-ttsmford-17} introduce a labeling concept.  Since the
horizontal and vertical directions on an ortho-radial grid are not
interchangeable (one is circular, the other is not), additional
information is required. The information which is the horizontal
direction is given by a \emph{reference edge} $\eref$ which is assumed
to lie on the outer face and that is directed such that it points in the
clockwise direction.  To present the characterization of Barth et
al.~\cite{bnrw-ttsmford-17}, we extend the notion of rotation.  For
two edges~$e,e'$ incident to a vertex $v$, let~$e=e_1,\dots,e_k=e'$ be
the edges between them in~$\mathcal E(v)$ so that~$(e_i,e_{i+1})$ is
an angle for $i=1,\dots,k-1$.  To measure the rotation between~$e$
and~$e'$, we convert the rotation values between them into geometric angles, sum
them up, and convert them back to a rotation, which gives
$\rot(e,e') = \sum_{i=1}^{k-1} \rot(e_i,e_{i+1}) - 2(k-2)$; see
Fig.~\ref{fig:prelim}b.  Note that for an angle $(e,e')$, it is
$k=2$, and therefore the two definitions of~$\rot(e,e')$ coincide.  For
a path~$P=v_0,\dots,v_{n-1}$ in $G$, we define its rotation
$\rot(P) = \sum_{i=1}^{n-1} \rot(v_{i-1} v_{i},v_iv_{i+1})$, and for a
cycle $C$ in $G$, we define its rotation
$\rot(C) =\sum_{i=1}^{n} \rot(v_{i-1} v_{i},v_iv_{i+1})$, where
indices are taken modulo~$n$; see Fig.~\ref{fig:prelim}c.  A cycle
$C$ of $G$ is called \emph{essential} if it separates the central and
the outer face.  A cycle is essential if and only if
$\rot(C) = 0$~\cite{bnrw-ttsmford-17}.  We assume that $C$ is directed
such that the central face lies to its right.  Let $e$ be an edge
on~$C$.  A \emph{reference path} for $e$ on $C$ is a (not necessarily
simple) walk~$P$ that starts with the edge $\eref$, ends with the edge
$e$ and does not contain an edge or a vertex that lies to the right of
$C$.  We define~$\ell_C(e) = \rot(P)$ as the label of $e$ on $C$; see
Fig.~\ref{fig:prelim}d.  Barth et al.~\cite{bnrw-ttsmford-17} show
that the label does not depend on the actual path~$P$ (however the
same edge may have different labels for different cycles).  With this,
Barth et al. formulate a third condition.  An ortho-radial
representation is called \emph{valid} if, for each essential cycle
$C$, either~$\ell_C(e) = 0$ for all edges $e \in E(C)$, or there exist
edges $e^-,e^+$ in $E(C)$ with~$\ell_C(e^-) < 0$ and~$\ell_C(e^+) > 0$.  A
cycle $C$ that does not satisfy this condition is called
\emph{strictly monotone}. Thus, an ortho-radial representation is
valid if and only if it has no strictly monotone cycle.  In
Fig.~\ref{fig:examples}a,b the edges of the 4-cycle are labeled with
their labels with respect to the reference edge $\eref$;
Fig.~\ref{fig:examples}b is a strictly monotone cycle.  The following
two results form the combinatorial and algorithmic basis for our work.

\begin{theorem}[Barth et al.~\cite{bnrw-ttsmford-17}]
  \label{thm:basic-combinatorial}
  An ortho-radial representation is drawable if and only if it is
  valid.
\end{theorem}

\begin{theorem}[Niedermann et al.~\cite{nrw-eaorg-19}]
  \label{thm:basic-algorithmic}
  There is an $O(n^2)$-time algorithm that, given an ortho-radial
  representation~$\Gamma$ of an $n$-vertex graph $G$, either outputs a
  drawing of~$\Gamma$, or a strictly monotone cycle $C$ in~$\Gamma$.
\end{theorem}

\section{ILP for Bend-Free Ortho-Radial Drawings}
\label{sec:bend-free}

\newcommand{\reverse}[1]{\overline{#1}}
\newcommand{\pred}{\text{pred}}

In this section we are given a planar 4-graph~$G=(V,E)$ with a
combinatorial embedding~$\mathcal E$, an outer face $f_o$, a central
face $f_c$ and a reference edge $\eref$ on $f_o$; we denote that
instance by $\mathcal G=(G,\mathcal E, f_o, f_c, \eref)$. We present
an algorithm based on an ILP that yields a valid ortho-radial
representation of $\mathcal G$, if it exists. %

\paragraph{Basic Formulation.}
For each vertex $u$ and each of its angles $(e,e')$ we introduce an
integer variable $r_{e,e'}\in \{-2,-1,0,1\}$, which describes the
rotation~$\rot(e,e')$ between $e$ and $e'$. Condition~\ref{condition:or:vertex}' is enforced by the following constraint for $u$.
\begin{eqnarray}
  \sum^k_{i=1} r_{e_{i},e_{i+1}} = 2(\deg(v)-2),\label{constr:or:vertex}
\end{eqnarray}
where $e_1,\dots,e_k$ are the incident edges of $u$ in
counter-clockwise order; we define $e_{k+1}=e_1$.  For each face $f$
of $\mathcal G$ Condition~\ref{condition:or:face}' is enforced by the
next constraint.
\begin{eqnarray}
  \sum_{i=1}^kr_{e_i,e_{i+1}} = \begin{cases}
    4  & \text{if $f$ is a regular face,}\\
    0  & \text{if $f$ is the outer or central face but not both,}\\
    -4 & \text{if $f$ is both the central and outer face,}
  \end{cases}\label{constr:or:face}
\end{eqnarray}
where $e_1,\dots,e_k$ are the edges of the facial walk around $f$ 
such that $f$ lies to the right; we define $e_{k+1}=e_1$.
We denote that formulation by $\mathcal F_{\mathrm{or}}$.  By
construction a valid assignment of the variables in
$\mathcal F_{\mathrm{or}}$ \emph{induces} an ortho-radial
representation~$\Gamma$. In particular, assuming that $e^\star$ is directed such that it points clockwise, we can derive from the variable assignment the
directions of the other edges in $G$. The next theorem summarizes this result.

\begin{theorem}\label{thm:ilp:or}
  An ortho-radial representation exists for $\mathcal G$ if and only if
  $\mathcal F_{\mathrm{or}}$ induces an ortho-radial representation.
\end{theorem}

\begin{figure}[t]
  \centering
  \includegraphics[page=1]{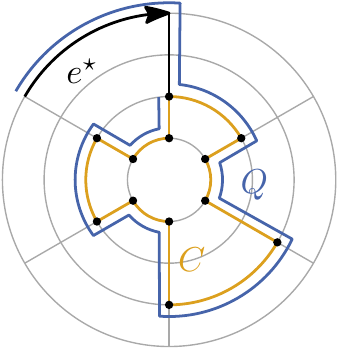}
  \caption{Illustration of path Q used for the labeling of $C$.}
  \label{fig:ilp:path-Q}
\end{figure}

However, the induced ortho-radial representation $\Gamma$ is not
necessarily valid, but may contain strictly monotone cycles. We
therefore extend $\mathcal F_{\mathrm{or}}$ by constraints for each
essential cycle $C$ of $\mathcal G$.  To that end, let $P$ be a path
that starts at $\eref$ and ends at $C$ such that it does not use any
vertex or edge that lies to the right of~$C$. Further, let $Q$ be the
path $\eref + P + C$ that follows $C$ in clockwise order from the
endpoint of $P$ and ends at the end point of $P$; see
Fig.~\ref{fig:ilp:path-Q}. For each edge $e$ of $Q$ we introduce an
integer variable $l_e$ with $-m \leq l_e \leq m$, which models a label
with respect to $C$. Here $m$ denotes the number of edges of $G$. We
require that the label of the reference edge is $0$, i.e.,
 $l_\eref = 0$. %
Moreover, for an edge $e=vw$ of $Q\setminus\{e^\star\}$ and its
predecessor $e'=uv$ on $Q$ let~$e'=e_1,\dots,e_k=e$ be the edges
between them in~$\mathcal E(v)$ so that~$(e_i,e_{i+1})$ is an angle
for $i=1,\dots,k-1$. We introduce the constraint
\begin{eqnarray}
  l_{e} = l_{e'}+ \sum_{i=1}^{k-1} r_{e_i,e_{i+1}} - 2(k-2)
\end{eqnarray}
Hence, the values $l_e$ for $e \in E(C)$ describe a labeling of
$C$, where $E(C)$ denotes the edges of $C$.  To exclude strictly
monotone cycles, we ensure that either~$l_e = 0$ for all edges
$e \in E(C)$, or there exist edges $e^-,e^+$ in $E(C)$
with~$l_{e^-} > 0$ and~$l_{e^+} < 0$. We first observe that $C$ can
only be strictly monotone if $\sum_{e \in E(C)} l_e\neq 0$. We
introduce a single binary variable $z$ that is $1$ if and only if
$\sum_{e\in E(C)} l_e = 0$. Additionally, for each edge $e$ of $C$ we
introduce two binary variables $x_e$ and $y_e$, which are used to enforce
that $l_e$ is negative or positive, respectively.

\begin{minipage}{0.45\textwidth}
  \begin{eqnarray}
  \sum_{e \in E(C)} l_e \leq  M \cdot(1 - z)\label{constr:mc:C} 
  \end{eqnarray}
\end{minipage}
\begin{minipage}{0.45\linewidth}
      \begin{eqnarray}
   \sum_{e \in E(C)} l_e \geq - M \cdot(1 - z)\label{constr:mc:D}
  \end{eqnarray}  
\end{minipage}

\begin{minipage}{0.45\textwidth}
  \begin{eqnarray}
  \sum_{e \in E(C)}x_e + z \geq 1\label{constr:mc:E}
  \end{eqnarray}
\end{minipage}
\begin{minipage}{0.45\textwidth}
  \begin{eqnarray}
   \sum_{e \in E(C)}y_e + z \geq 1\label{constr:mc:F}
  \end{eqnarray}
\end{minipage}

\hspace{-6ex}
\begin{minipage}{0.5\textwidth}
  \begin{eqnarray}
   l_e \leq -1 + M\cdot (1-x_e)  \ \forall e\in E(C)\label{constr:mc:A}
  \end{eqnarray}
\end{minipage}
\begin{minipage}{0.5\textwidth}
  \begin{eqnarray}
  l_e \geq 1 - M\cdot (1-y_e)  \ \forall e\in E(C)\label{constr:mc:B}
  \end{eqnarray}
\end{minipage}
\vspace{2ex}

We define $M$ as a constant with $M>m$ so that the corresponding
constraints are trivially satisfied for $z=0$, $x_e=0$ and $y_e=0$,
respectively. If $z=1$, we obtain by Constraint~\ref{constr:mc:C} and
Constraint~\ref{constr:mc:D} that $\sum_{e\in E(C)} l_e = 0$. Hence,
$C$ is not strictly monotone. Otherwise, if $z=0$, by
Constraint~\ref{constr:mc:E} there is an edge $e^-\in E(C)$ with
$x_{e^-}=1$. By Constraint~\ref{constr:mc:A} we
obtain~$l_{e^-}<0$. Similarly, by Constraint~\ref{constr:mc:F} there
is an edge $e^+\in E(C)$ with $y_{e^+}=1$. By
Constraint~\ref{constr:mc:B} we obtain~$l_{e^+}>0$. Altogether, we
find that $C$ is not strictly monotone.  We emphasize that for each
essential cycle $C$ of $\mathcal G$ we introduce a fresh set of variables
and constraints; which we denote by $\mathcal F_C$. Hence, we
consider the ILP
$\mathcal F(\mathcal G)=\mathcal F_{\mathrm{or}} \cup \bigcup_{C\in \mathcal
  C} \mathcal F_C$, where $\mathcal C$ is the set of all essential
cycles in $\mathcal G$. The next theorem summarizes this.
\begin{theorem}\label{thm:valid-or}
  If $\mathcal G$ has an ortho-radial representation, then the
  formulation $\mathcal F(\mathcal G)$ induces a valid ortho-radial
  representation.
\end{theorem}

\paragraph{Separation of Constraints.}
Adding $\mathcal \mathcal F_C$ for each essential cycle $C$ of
$\mathcal G$ is not feasible in practice, as there can be
exponentially many of these in $\mathcal G$.  Hence, instead, we
propose an algorithm that adds $\mathcal F_C$ on demand. The algorithm
first checks whether $\mathcal G$ has an ortho-radial representation
$\Gamma_1$ using the formulation $\mathcal F_1:=\mathcal F_{or}$
(Theorem~\ref{thm:ilp:or}). If this is not the case, the algorithm
stops and returns that there is no ortho-radial representation for
$\mathcal G$. Otherwise, starting with $\mathcal F_1$ and $\Gamma_1$
it applies the following iterative procedure. In the $i$-th iteration
(with $2\leq i$) it checks whether $\Gamma_{i-1}$ is valid
(Theorem~\ref{thm:basic-algorithmic}). If it is, the algorithm stops
and returns $\Gamma_{i-1}$. Otherwise, the validity test yields a
strictly monotone cycle~$C$ as a certificate proving that $\Gamma_{i-1}$
is not valid. The algorithm creates then the formulation
$\mathcal F_{i}=\mathcal F_{i-1}\cup \mathcal F_C$ and induces the
ortho-radial representation~$\Gamma_i$, in which it is enforced that
$C$ is not strictly monotone. The algorithm stops at the latest when
the formulation $\mathcal F_C$, which prohibits that $C$ is a strictly
monotone cycle, has been added for each essential cycle
$C\in \mathcal C$. Hence, in theory an exponential number of
iterations may be necessary. However, in our experiments the procedure
stopped after few iterations for all test instances; see
Section~\ref{sec:evaluation}.

\paragraph{Bend Optimization.}
We also can use the ILP to optimize the ortho-radial
representation. In Section~\ref{sec:with-bends} we consider bend
minimization by modeling bends as degree-2 vertices. We therefore
extend $\mathcal F_{\mathrm{or}}$ such that it allows us to optimize
the change of direction at such nodes. For each degree-2 vertex we
introduce a binary variable $c_u$, which is $1$ if and only if one of
the two incident edges of $u$ lies on a concentric circle and the
other lies on a ray of the grid. The two incident edges $e_1$ and
$e_2$ of $u$ form the two angles $(e_1,e_2)$ and $(e_2,e_1)$. For
these we introduce the constraints $c_u \geq r_{e_1,e_2}$ and
$c_u \geq r_{e_2,e_1}$. Subject to these constraints we minimize
$\sum_{u \in V_2}c_u$, where $V_2\subseteq V$ denotes the degree-2
vertices of $G$. We can easily restrict the optimization to a subset
of $V_2$ distinguishing between degree-2 vertices
that originally belong to $G$ and those that we use for modeling
bends.

\section{Optimizing Bends and the Choice of the Central Face}
\label{sec:with-bends}

\begin{figure}[t]
  \centering
  \includegraphics[page=2]{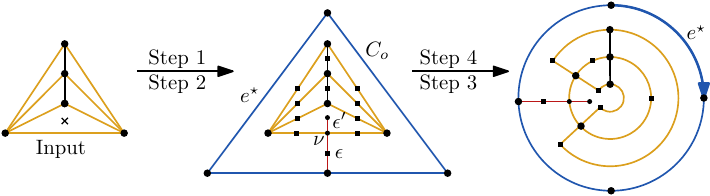}
  \caption{Illustration of the layout algorithm. Subdivision vertices are squares.}
  \label{fig:algorithm:four-steps}
\end{figure}

In this section we are given a graph $G$ with embedding $\mathcal E$
and designated outer face $f_o$. We describe an algorithm that returns
a bend-optimal ortho-radial drawing for
$\mathcal G=(G,\mathcal E, f_o)$, i.e., there is no other ortho-radial
drawing $\mathcal G$ that has fewer bends for any choice of the
central face $f_c$ and the reference edge $e^\star$ on $f_o$.  The
algorithm uses the ILP from Section~\ref{sec:bend-free} as a building
block.  The ILP does not directly allow to express bends; rather, we
subdivide the edges with degree-2 vertices, which can then be used as
bends. See Fig.~\ref{fig:algorithm:four-steps} for an illustration.

\begin{enumerate}
\item Insert a cycle $C_o$ in $\mathcal E$ that encloses $G$ and
  connect $C_o$ via an edge~$\epsilon$ to a newly inserted vertex~$\nu$ on
  the original outer face of $\mathcal E$. Insert edge $\epsilon'$ on
  the opposite side enforcing that $\nu$ has degree 4. Hence, $C_o$ is the
  new boundary of the outer face $f_o$. Choose an arbitrary edge of
  $C_o$ as reference edge~$e^\star$.
\item Subdivide each edge of $G$ with degree-2 vertices such that each
  maximally long chain of degree-2 vertices consists of at least $K$ vertices.
\item Create a valid ortho-radial representation~$\Gamma_f$ for face
  $f$ as the central face. To that end, apply the ILP formulation of
  Section~\ref{sec:bend-free} with separated constraints and bend
  optimization on $\mathcal G_f=(G,\mathcal E, f_o, f,e^\star)$
  charging the newly inserted degree-2 vertices with costs; the
  subdivision vertices on $\epsilon$ are not charged.
\item For the representation $\Gamma_f$ with fewest bends compute a
  drawing (Theorem~\ref{thm:basic-algorithmic}).
\end{enumerate}
  \begin{figure}[tb]
    \centering
    \includegraphics[page=1]{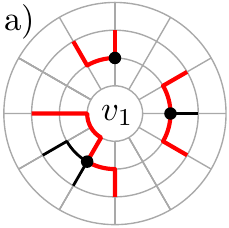}\hfil
    \includegraphics[page=2]{fig/bk-steps}\hfil
    \includegraphics[page=3]{fig/bk-steps}\hfil
    \includegraphics[page=4]{fig/bk-steps}
    \caption{Constructions for the proof of Theorem~\ref{thm:bounds}}
    \label{fig:bk-steps:template}
  \end{figure}
For bend-optimal drawings an appropriately large $K$ is decisive.  For
biconnected graphs $K=2n+4$ is sufficient even for a fixed central and outer face. 
\begin{theorem}\label{thm:bounds}
  Every biconnected plane 4-graph on $n$ vertices with designated
  central and outer faces has a planar ortho-radial drawing with at
  most $2n+4$ bends and at most two bends per edge with the exception
  of up to two edges that may have three bends.
\end{theorem}
The proof, which is deferred to Appendix~\ref{apx:omitted-proof}, uses similar
constructions as in the
orthogonal case; see also Fig.~\ref{fig:bk-steps:template}. It seems plausible that the bound can be transferred
to non-biconnected graphs as in the work by Biedl and
Kant~\cite{Biedl1998}; as we use a different bound, we refrain from the
rather technical proof. Moreover, we insert $C_o$ to make the layout
independent from the choice of the reference edge $\eref$. This does
not impact the number of bends needed for the original part
of $G$, because we subdivide $\epsilon$ with sufficiently many 2-degree vertices
that can be bent for free. Further, as $\nu$ has degree 4, the
drawing cannot be bent at $\nu$.

Replacing each edge by $2n+4$ degree-2 vertices increases the size of
the graph drastically. However, the ILP can be solved much faster if
fewer subdivision vertices are used. Next, we describe a
pruning strategy that uses upper and lower bounds on the optimal
drawing to exclude central faces and to limit the number of
subdivision vertices and the number of times we solve the ILP.

We first compute the minimum number $U$ of bends that is necessary for
a bend-optimal orthogonal drawing of $\mathcal G$. This also bounds
the number of bends in a bend-optimal ortho-radial drawing of
$\mathcal G$. Hence, it is sufficient to subdivide each edge with $U$
vertices in Step~2. Initially, we run $\mathcal F_\mathrm{or}$ on each face $f$ of $\mathcal G$ as central face. This gives us a
lower bound $l_f$ for the bends in the case that $f$ is the central
face. In Step 3 we then consider the faces in increasing order of
their lower bounds. If the lower bound $l_f$ of the current face $f$
exceeds the upper bound $U$ we prune $f$ and continue with the next
face. Otherwise, we iteratively compute a valid ortho-radial
representation $\Gamma_f$ for $f$ as described in
Section~\ref{sec:bend-free} and update $U$ if it is improved by the
current solution. Further, when we update the ILP due to strictly
monotone cycles, we skip $f$ if its number of bends exceeds $U$ and
continue with the next face.

\section{Experimental Evaluation}
\label{sec:evaluation}
In this section we present our experimental evaluation which we have
conducted to show the potential of our approach as a general graph
drawing tool.

\begin{figure}[t]
  \centering
  \includegraphics[width=0.8\linewidth]{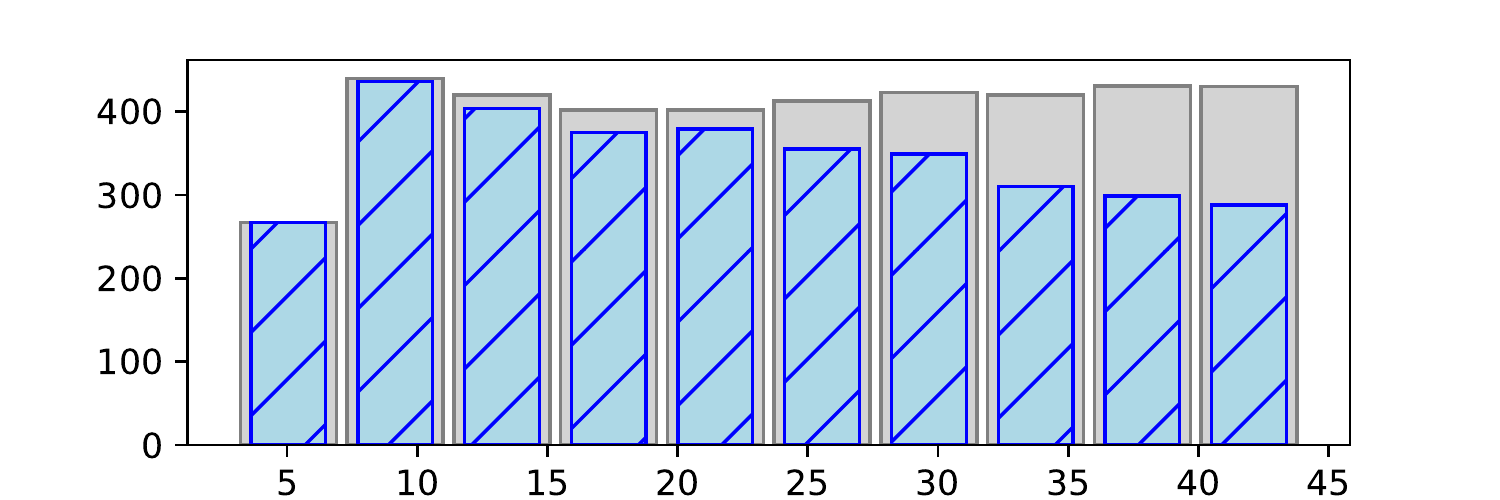}
  \caption{The node distribution (gray bars) of the graphs in $\mathcal I_\text{Rome}$
    distributed on 10 equally sized bins. The number of vertices ranges between $\minNumVertices$ and $\maxNumVertices$. The blue, tiled bars indicate the number of optimally solved instances.}
  \label{fig:eval:histogram}
\end{figure}

\subsection{Feasibility of Approach}
We first pursue the issue of whether our approach is feasible. It is
far from clear whether prohibiting strictly monotone cycles on demand
is practical, as we may need to insert an exponential number of
constraints into the ILP formulation.  To answer this question we have
conducted the first experiments on a subset of the \emph{Rome
  graphs}\footnote{\url{http://www.graphdrawing.org/data}},
which is a widely accepted benchmark set. We have replaced
  each vertex~$v$ with degree~$k>4$ with a cycle of $k$ vertices,
  which we connected to the neighbors of $v$ correspondingly.
Further, we applied a heuristic from OGDF~\cite{CGJKKM2013} to embed
the remaining graphs such that the size of the outer face is
maximized. We replaced all edge crossings with
  degree-4 vertices.  A preliminary analysis showed that the graphs
contain many degree-2 vertices. To ensure for the purpose of the
evaluation that our approach is forced to introduce bends with costs,
we normalized each instance by removing all degree-2 vertices. We only
considered instances up to 44
nodes. In total we obtained a set $\mathcal I_\text{Rome}$ of
\numInstances instances. Figure~\ref{fig:eval:histogram} shows the
size distribution of the resulting instances.  We implemented our
approaches in Python and solved the ILP formulations using
Gurobi~9.0.2~\cite{gurobi} using a timeout of 2 minutes in each
iteration. We ran the experiments on an Intel(R) Xeon(R) W-2125 CPU
clocked at 4.00GHz with 128 GiB RAM.

\newcommand{\romeGraph}[2]{\begin{minipage}{#2\textwidth}\includegraphics[width=\textwidth, page=#1]{fig/romegraphs}\end{minipage}}
\begin{figure}[t]
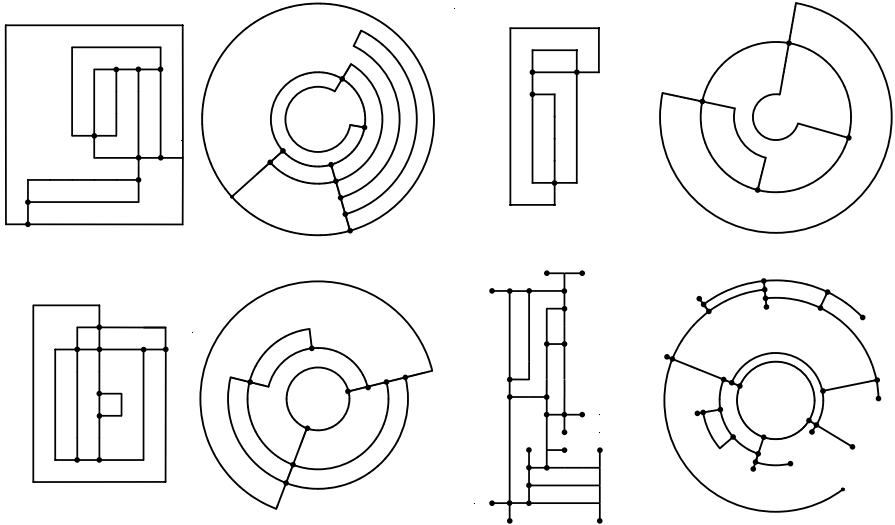

  \centering
  \romeGraph{1}{0.49}  \romeGraph{2}{0.49}\vspace{1ex}
  
  \romeGraph{3}{0.49}  \romeGraph{4}{0.49}
  \caption{Examples of
    bend-optimal orthogonal and ortho-radial drawings for the Rome graphs. The outer face was fixed, but the central face was optimized.  }
  \label{fig:examples:rome-graphs}
\end{figure}

For each of the instances in $\mathcal I_\text{Rome}$ we applied the
algorithm described in Section~\ref{sec:with-bends}; see
Fig.~\ref{fig:examples:rome-graphs} for four examples.  For
\optimallySolved instances we obtained bend-optimal ortho-radial
drawings. For \timeOut instances the solver returned a not
necessarily optimal result due to timeouts. The number of not optimally solved instances increases with the number of nodes; see Fig.~\ref{fig:eval:histogram} for more details.  For \timeLessHalfSecond
instances the algorithm took less than half a second. Only for
\timeMoreTenSeconds instance it took more than 10 seconds;
\timeMoreOneMinute of them took more than one minute. Further, when
searching for the best choice of the central face about
$\excludedFacesAvg\%$ of the faces are pruned in advance on average
due to exceeding upper bounds. Hence, for more than three quarters of the
faces we do not need to solve the ILP formulation, still guaranteeing
that we obtain a drawing with minimum number of bends.  Moreover, when the algorithm runs for a fixed central face, it needs less than
\numIterationsMaxAvg iterations on average until it finds a valid
ortho-radial representation. Put differently, we insert the
formulation $\mathcal F_C$ prohibiting a strictly monotone cycle $C$
into the ILP formulation \numIterationsMaxAvg times on
average. Altogether, the evaluation shows the practical feasibility of
the approach. It supports the rather strong hypothesis that
prohibiting strictly monotone cycles on demand is sufficient, but
considering all essential cycles is not necessary.

\begin{figure}[t]
  \begin{picture}(170,100)
\put(0,0){\includegraphics[width=0.49\textwidth]{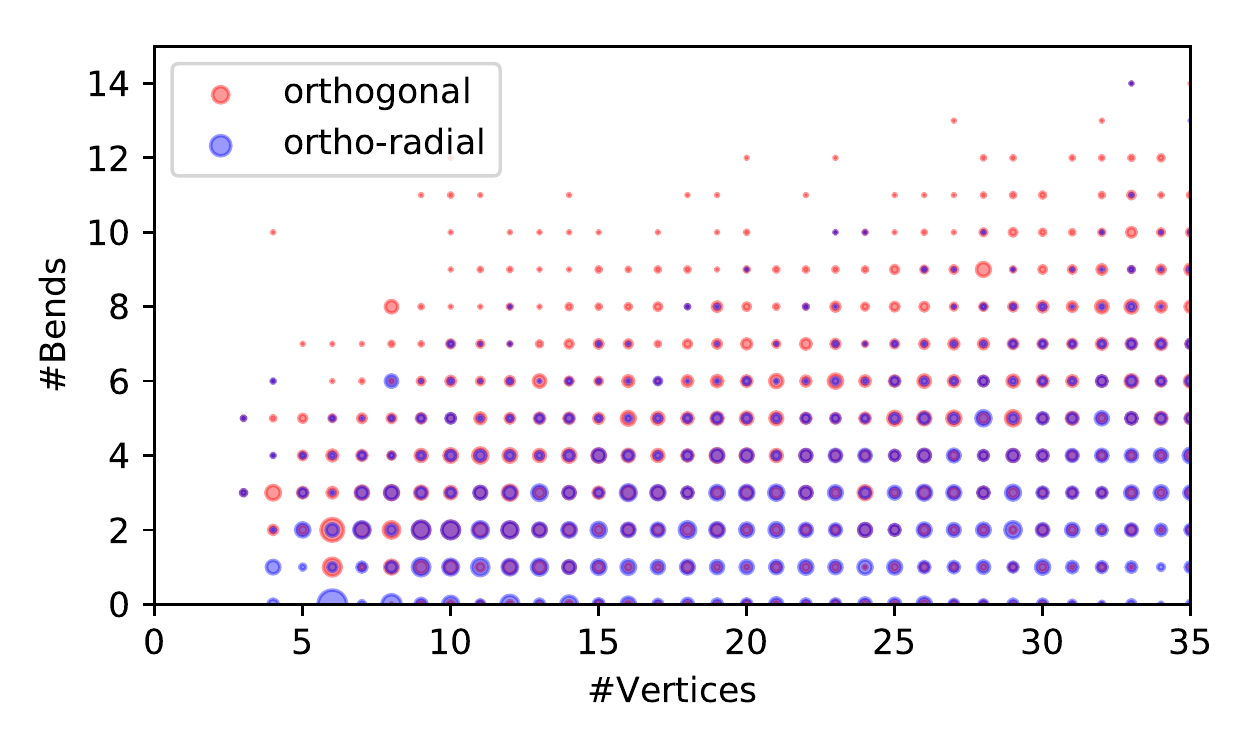}}
\put(0,90){a)}
\end{picture}
\begin{picture}(100,100)
\put(0,0){\includegraphics[width=0.49\textwidth]{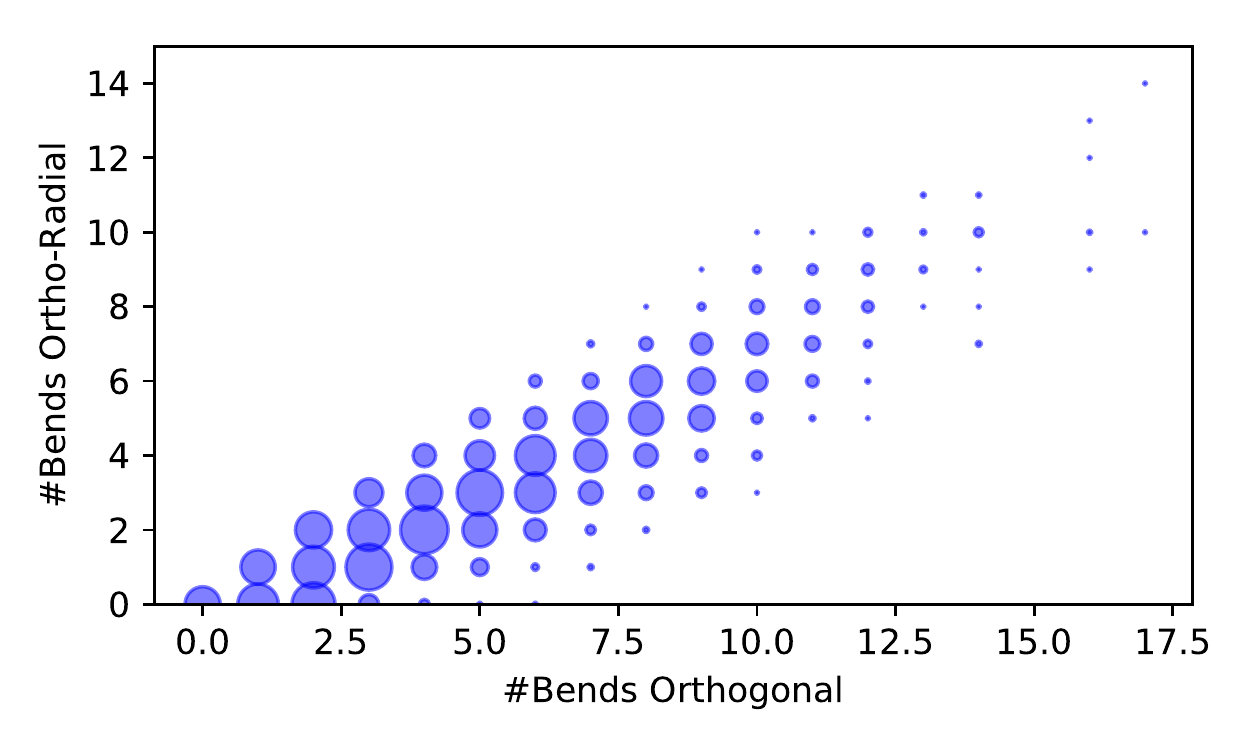}}
\put(0,90){b)}
\end{picture}
 \caption{Overview of the considered Rome graphs. A disk with radius $r$ and position $(x,y)$ corresponds to $r$ instances (a)~with $x$ vertices and $y$
   bends in the ortho-radial (blue) and the orthogonal (red) drawing, (b)~with $x$ bends in the orthogonal drawing and $y$ bends in the ortho-radial drawing.}
 \label{fig:plots}
\end{figure}

\subsection{Ortho-radial Drawings vs.~Orthogonal Drawings}
In this part we compare ortho-radial drawings with orthogonal drawings
with respect to the necessary number of bends. We expect a reduction
of the number of bends in an ortho-radial drawing compared with its
orthogonal drawing.

Fig.~\ref{fig:plots}a shows that independent of the size of
the graphs the ortho-radial drawings often have fewer bends than the
orthogonal drawings. Further, Fig.~\ref{fig:plots}b shows that
for many of the instances we achieve a reduction between 1 to 3 bends
in the ortho-radial drawings. To investigate this in greater detail we
consider for each instance $I\in \mathcal I_{\text{Rome}}$  the \emph{bend
  reduction}
 $ r_I = \frac{b_{\mathrm{og}}-b_{\mathrm{or}}}{b_{\mathrm{og}}}\cdot 100\%$,
 where $b_{\mathrm{og}}$ is the minimum number of bends of an
 orthogonal drawing of $I$ and $b_{\mathrm{or}}$ is the number of
 bends of the ortho-radial drawing created with our approach; note
 that for both drawings we assume the same embedding and the same
 outer face.  From this comparison we have excluded any instance with
 zero bends.  The bend reduction is $\reductionAvg\%$ on average and
 the median is at $\reductionMedian\%$. We emphasize that for
 \orZeroBends instances there are bend-free ortho-radial drawings,
 whereas only \orthogonalZeroBends admit bend-free orthogonal
 drawings.  Thus, our experiments support our hypothesis that
 ortho-radial drawings lead to a substantial bend reduction.

\begin{figure}[t]
  \centering
  \subfloat{\includegraphics[page=2,width=0.48\textwidth]{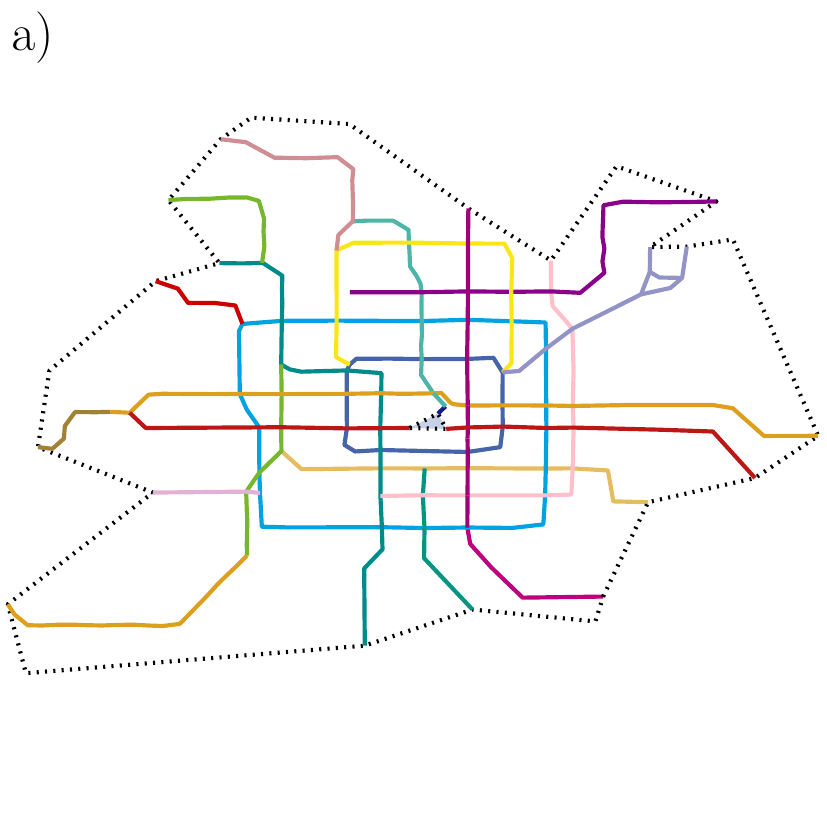}}\hfill
 \subfloat{\includegraphics[page=3,width=0.48\textwidth]{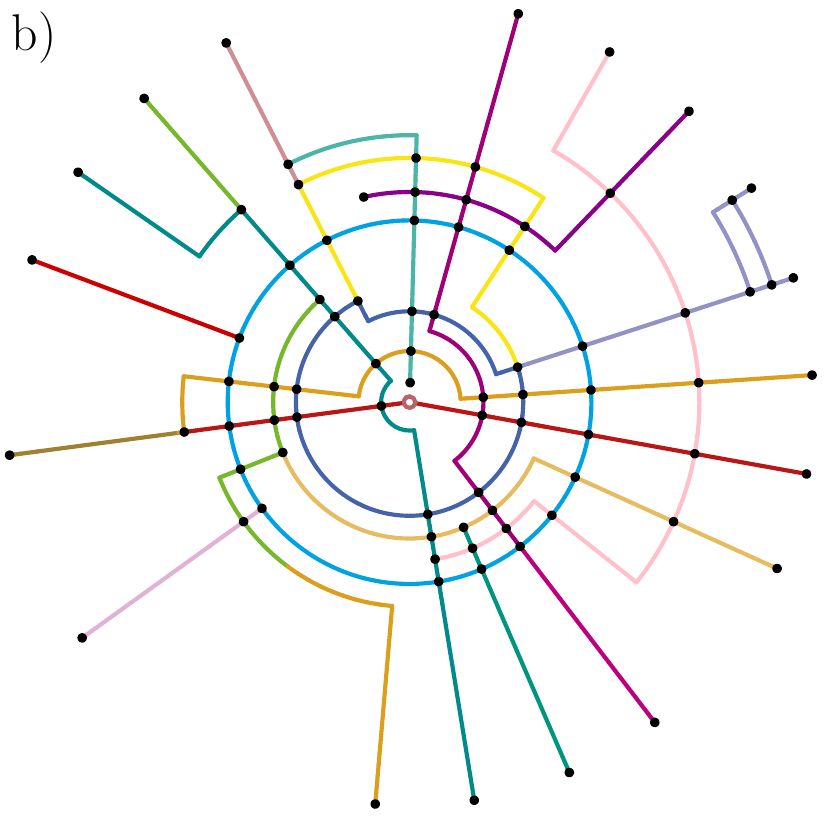}}
 \caption{The metro system in Beijing, China. (a)~The input graph
   derived from vectorizing a metro map of Beijing. The outer and
   central faces are dashed. (b)~The ortho-radial layout induced by
   our approach within 7 seconds. }
  \label{fig:metro-map:beijing}
\end{figure}

\subsection{Case Study on Metro Maps}
Ortho-radial drawings are particularly used to represent metro
systems~\cite{wntrn-stlfdmhp-20}. We tested our algorithm on the metro
system of Beijing, which is a comparably large and complex transit
system; see Fig~\ref{fig:metro-map:beijing}. We have vectorized a
metro map of the city that shows 21 lines; for details see
Appendix~\ref{apx:vectorization}.  The created graph has 224 vertices,
289 edges and 67 faces. We fixed the central face by hand to
intentionally determine the appearance of the final layout.
We subdivided the edges such that each chain consisting of degree-2
vertices has at least three intermediate vertices. Our algorithm
created the layout shown in Fig.~\ref{fig:metro-map:beijing}b within
seven seconds. It has 21 bends. We emphasize that the outer loop line
is represented as a circle and the inner loop line has only two
bends. Altogether, the layout reflects the main geometric features of
the system well, although we have only optimized the number of bends,
e.g., outgoing metro lines are mainly drawn as straight-lines
emanating from the center.  In a second run, which took three minutes,
we proved that 21 bends is optimal.  Further metro systems are
found in Appendix~\ref{apx:examples}.

\section{Conclusion}
Barth et al.~\cite{bnrw-ttsmford-17} and Niedermann et
al.~\cite{nrw-eaorg-19} carried over the metrics step of the TSM
framework from orthogonal to ortho-radial drawings explaining how to
obtain such a drawing from a valid ortho-radial
representation. However, they let open how to transfer the shape step
constructing such a valid ortho-radial representation. We presented
the first algorithm that answers this question and creates
ortho-radial drawings, which are bend-optimal. Our
experiments showed its feasibility based on the Rome graphs and
different metro systems. This was far from clear due to the possibly
exponential number of essential cycles.

Altogether, we presented a general tool for creating ortho-radial
drawings.  We see applications in map making (e.g., metro
maps, destinations maps).  Possible future refinements include the
adaption of the optimization criteria both in the shape and metrics
step. For example in the shape step one could enforce certain bends to
better express the geographic structure of the transit
system. %

\bibliographystyle{splncs04}
\bibliography{orthoradial-ilp-full}

\appendix

\newpage

\begin{figure}[t!]
  \centering
  \includegraphics[page=1,scale=0.8]{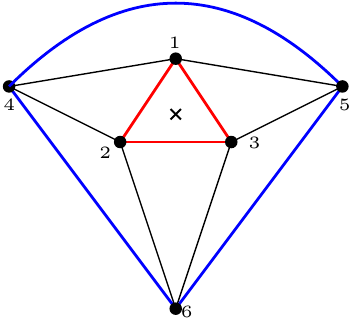}

  \includegraphics[page=2,scale=0.8]{fig/bk-algo} \hfil \includegraphics[page=3,scale=0.8]{fig/bk-algo}

  \includegraphics[page=4,scale=0.8]{fig/bk-algo} \hfil \includegraphics[page=5,scale=0.8]{fig/bk-algo}

  \includegraphics[page=6,scale=0.8]{fig/bk-algo} \hfil \includegraphics[page=7,scale=0.8]{fig/bk-algo}
  \caption{Example of the algorithm from Theorem~\ref{thm:bounds} applied to the
    octahedron graph.  The central face is marked with a cross, the
    numbers indicate the st-ordering that was used.}
    \label{fig:octahedron}
\end{figure}

\section{Omitted Proof of Section~\ref{sec:with-bends}}
\label{apx:omitted-proof}
\begin{theorem}
  Every biconnected plane 4-graph on $n$ vertices with designated
  central and outer faces has a planar ortho-radial drawing with at
  most $2n+4$ bends and at most two bends per edge with the exception
  of up to two edges that may have three bends.
\end{theorem}

\begin{proof}
  Let $G$ be a biconnected 4-plane graph and let $s,t$ be two vertices
  incident to the central and outer face, respectively, that are not
  adjacent.  Let $\langle s=v_1,\dots,v_n=t\rangle$ be an st-ordering
  of $G$.

  \begin{figure}[tb]
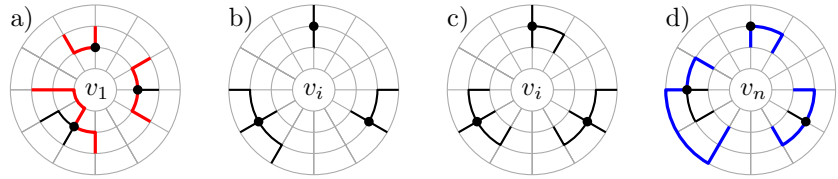

    \centering
    \includegraphics[page=1]{fig/bk-steps}\hfil
    \includegraphics[page=2]{fig/bk-steps}\hfil
    \includegraphics[page=3]{fig/bk-steps}\hfil
    \includegraphics[page=4]{fig/bk-steps}
    \caption{Illustration of the incremental drawing step.  Drawing of
      the first vertex $v_1$, depending on its degree (a); the edges
      incident to the central face are thick and red.  Drawing of the
      intermediate vertices depending on their in- and outdegrees
      (b,c).  Drawing of the last vertex $v_n$ depending on its degree
      (d); the edges incident to the outer face are thick and blue.}
    \label{fig:bk-steps}
  \end{figure}

  Our construction iteratively places each vertex $v_i$ onto the
  circle with radius~$i$ around the origin; the construction is
  illustrated in Figure~\ref{fig:octahedron}. Edges where both
  endpoints have already been placed are \emph{drawn}, edges where
  exactly one endpoint has already been placed, are \emph{partial}.
  Note that, when inserting an edge, at least one partial edge becomes
  drawn, and at most three edges become partial.

  Throughout, we maintain the invariant that (i) all drawn edges
  (except at most one incident to $v_1$ and~$v_n$, each) have at most
  two bends , (ii) they are contained in the disk with radius $i$, and
  (iii) the half-edges are drawn as stubs with at most one bend such
  that (iv) only an outward-directed segment lies outside the disk
  with radius $i$, and these end at the circle with radius $i+0.5$.
  Moreover, since we have an $st$-ordering, for each $i<n$ the
  vertices $v_1,\dots,v_i$ and the vertices $v_{i+1},\dots,v_n$ induce
  connected subgraphs.  Therefore the cut $C$ separating them is
  represented by a simple curve in the dual.  We further require (v)
  that our drawing respects the planar embedding in the sense that the
  circular order of the stubs around this circle is the same as the
  order in which the they are intersected by the cycle that is dual
  to~$C$.

  Depending on the degree of~$v_1$, we draw it together with its
  half-edges as illustrated in Fig.~\ref{fig:bk-steps}a; where we
  choose the directions in which the edges leave $v$ so that the edges
  incident to the central face are drawn as indicated by the thick
  blue curves.  This clearly establishes properties (i)--(v).  Suppose
  $1<i<n$.  Let $u_1,\dots,u_k$ with $1\le k \le 4$ be the neighbors
  of $v_i$ that are already drawn.  By property~(v) and the fact that
  we have an st-ordering, it follows that the ends of the half-edges
  from the $u_j$ to $v_i$ are consecutive around the circle with
  radius $i-0.5$.  Depending on the in- and out-degree, we position
  $v$ in the outward direction above one of its incoming edges and
  draw the outgoing edges as illustrated in
  Fig.~\ref{fig:bk-steps}b,c, where the spokes that contain the
  outgoing edges of $v$ are newly created left and right of the spoke
  that contain~$v_i$, and the remaining spokes are slightly squeezed
  to make sufficient space.  All remaining stubs are simply extended
  by one unit in the outward direction.  Finally, we palce the vertex
  $v_n$ as illustrated in Fig.~\ref{fig:bk-steps}d; making sure that
  it is positioned in the outward direction above one of its incoming
  edges in such a way that the correct faces lies on the outside.

  Clearly each edge, except for one edge incident to $v_1$ and~$v_4$
  receive at most two bends (one when its first vertex is drawn, and
  one when the second vertex is drawn.  Moreover, each vertex causes
  bends on at most two of its incident vertices, which yields an upper
  bound of at most $2n$ bends.  Moreover, the special edges incident
  to $v_1$ and~$v_n$ receive two additional bends, which yields the
  claimed total of $2n+4$.
\end{proof}

\section{Metro Maps}

\subsection{Vectorization of the Metro System of Beijing}\label{apx:vectorization}
We vectorized the metro system of Beijing as follows. 
We connected crossings between metro lines and terminal stations by
paths consisting of degree-2 vertices. We refrained from modeling the
intermediate stations as degree-2 vertices; one may distribute them on
the sections of the metro lines after creating the ortho-radial
layout. The metro system has only one degree-5 vertex. We resolved
this vertex by reconnecting one of the five edges to one additional
vertex subdividing a neighboring edge.  Further, we have replaced the
station Tiananmen East, which lies in the center of the city, by a
cycle of three edges. We have fixed this cycle as the boundary of the
central face. Further, we have connected the terminal stations of the
outgoing lines by a cycle enclosing the entire system; we fixed this
cycle as the boundary of the outer face.  The resulting graph has 224
vertices, 289 edges and 67 faces.

\subsection{Additional Examples}\label{apx:examples}
In Figure~\ref{fig:metro-map:cologne} and
Figure~\ref{fig:metro-map:london} we present ortho-radial layouts for
the metro systems of Cologne, Germany and London, UK. We extracted the
graphs to obtain large examples for our feasibility study. In
particular, we do not claim that the layouts correctly represent the
transit systems. Especially for the London metro system we resolved
several degree-5 vertices modeling them as cycles. Hence, as we do not
impose any restrictions on the layout, geographically close stations
may be placed far away from each other in the layout. We deem the task
of transferring the algorithm to produce reliable metro maps to be an
engineering problem that should be tackled in future work.
\begin{figure}[t]
  \centering
  \subfloat{\includegraphics[page=1,width=0.48\textwidth]{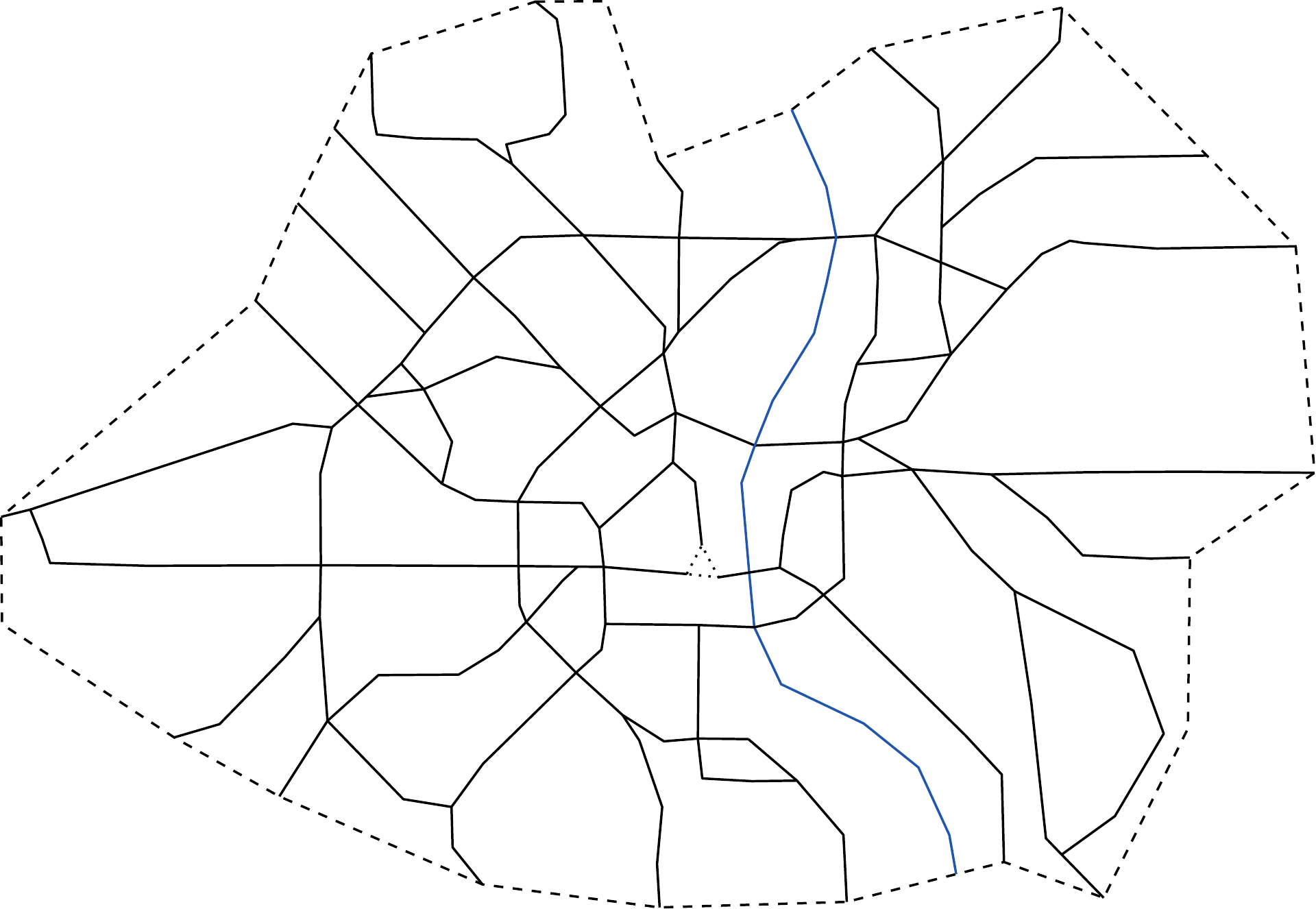}}\hfill
 \subfloat{\includegraphics[page=2,width=0.48\textwidth]{fig/cologne}}
 \caption{The metro system of Cologne, Germany. The river Rhine, which
   passes through the city, is marked blue. (a)~The input graph
   derived from vectorizing the official metro map of Cologne. We
   replaced the station \emph{Heumarkt}, which lies in the center of
   the city, with a cycle and fixed this as the central face. The
   outer and central faces are dashed. The graph has 177 vertices, 231
   edges and 56 faces.  (b)~The ortho-radial layout produced by our
   approach within 2 seconds. It has 18 bends. }
 \label{fig:metro-map:cologne}
   \subfloat{\includegraphics[page=1,width=0.48\textwidth]{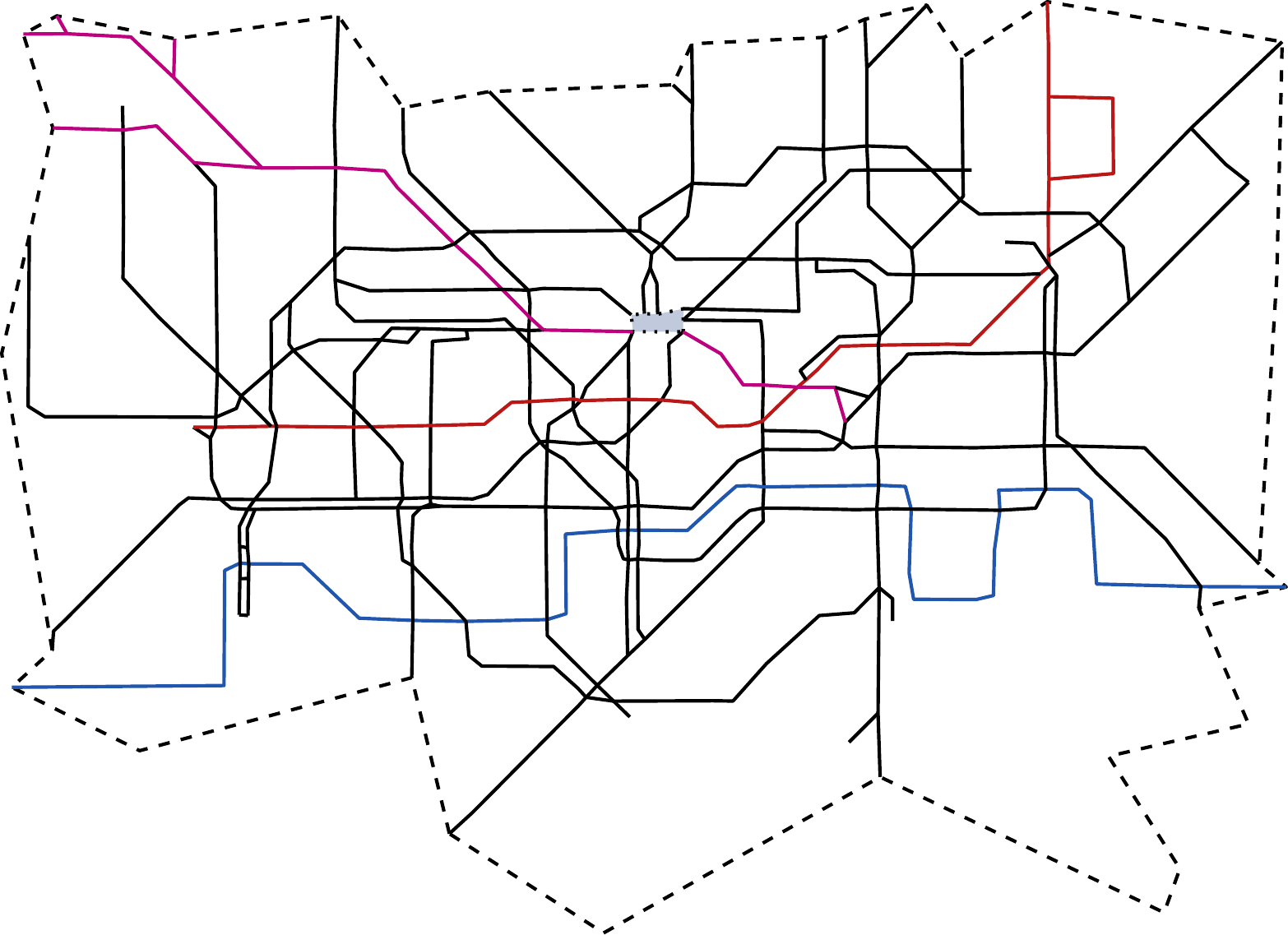}}\hfill
 \subfloat{\includegraphics[page=2,width=0.48\textwidth]{fig/london}}
 \caption{The metro system of London, UK. We have marked two lines by
   their color and the river Thames (blue) to give
   orientations. (a)~The input graph derived from vectorizing the
   official metro map of London. We modeled the area around
   \emph{King's Cross} with a cycle and fixed this as the central
   face. The outer and central faces are dashed. The graph has 398
   vertices, 530 edges and 134 faces. (b)~The ortho-radial layout
   produced by our approach within 9 seconds. It has 62 bends.}
  \label{fig:metro-map:london}
\end{figure}

\end{document}